\documentclass{article}

\usepackage{arxiv}

\usepackage[utf8]{inputenc} 
\usepackage[T1]{fontenc}    
\usepackage{url}            
\usepackage{booktabs}       
\usepackage{amsfonts}       
\usepackage{nicefrac}       
\usepackage{microtype}      
\usepackage{lipsum}

\usepackage{relsize,balance,lipsum,bbm,enumerate,times,comment,color,graphicx,setspace,mathdots,mathrsfs,amssymb,latexsym,amsfonts,amsmath,cite,stmaryrd,caption,pgf,accents,mathtools,tabu,enumitem,hhline,array,epstopdf,nicefrac,amsthm,microtype,algorithmic,array,float,bm,url}

\usepackage{graphicx}

\usepackage{mathtools} 

\usepackage[utf8]{inputenc}
\usepackage[english]{babel}

\usepackage{tikz}
\usetikzlibrary{positioning}
\definecolor{mygreen}{RGB}{153,255,153}
\definecolor{myorange}{RGB}{255,178,102}
\definecolor{myred}{RGB}{255,153,153}
\definecolor{myblue}{RGB}{153,204,255}

\title{Characterization of Potential Games: Application in Aggregative Games}

\author{
    Sina~Arefizadeh*\\
    Dept. of Electrical and Computer Engineering\\
    Arizona State University\\
    Tempe, Arizona, USA\\
    \texttt{sarefiza@asu.edu}
\And
    Angelia~Nedi\'c\\
    Dept. of Electrical and Computer Engineering\\
    Arizona State University\\
    Tempe, Arizona, USA\\
    \texttt{Angelia.Nedich@asu.edu}
}

\theoremstyle{definition}
\newtheorem{lemma}{Lemma}
\newtheorem{thm}{Theorem}

\newtheorem{definition}{Definition}
\newtheorem{rem}{Remark}

\begin{document}
\maketitle
\begin{abstract}
The main objective of this work is to describe games which fall under title of Potential and simplify the conditions for class of aggregative games. Games classified as aggregative are ones in which, in addition to the player's own action, the payoff for each player depends on an aggregate of all the players' decision variables. In this study, we developed a method based on payoff functions to determine if a given game is potential. Then, in order to identify the Aggregative Games that fall under this class we simplified the criteria for the class of Aggregative Games. A $3$-player Cournot game, also known as an Aggregative Potential Game, is used to test the characterization criteria for Potential Games. A $4$-player Cournot game is also utilized to test the form of potential function we obtained for class of general potential games. 

\end{abstract}


\section{Introduction}

Study of potential games has long been topic of interests due to the nice equilibrium behavior of this class of games. However, prior to this study there were no focus on understanding the behavior of this class of games and generally characterizing them regardless of payoff continuity and differentiablity. After a complete characterization of class of potential games, we can check particular classes of games that under what conditions they fall under potential games. For this aim we proposed a necessary and sufficient condition for this class of potential games in terms of sub game with every two player involved in the game. Then we stepped forward to simplify this criteria for class of aggregative potential games.

In the category of non-cooperative games, each player's payoff function may be influenced by the other agents' decisions. The payoff functions, however, in a wide range of games depend on some function of the aggregate decision variable of all agents. As opposed to other agents' individual supply, the payoff function of players in the Cournot game depends on the overall supply of opponents. In recent years, this category of games has attracted interest from a wide range of fields, including electrical engineering and transportation science. References \cite{C4,C5,C7,C8,C9,C10,C11} are examples of such studies.

The concept of anonymous player interaction is the cornerstone of the idea of aggregative games. This feature greatly simplifies the game for the participants and the study of the game for those who employ game theory as a tool. To establish my own decision variable, I just need to forecast the aggregate decision of all other agents from the perspective of the player. From the standpoint of an analyst, it is not necessary to concentrate on predicting what will happen with every possible value of a decision variable used by every other agents.

In the study of games, the Nash equilibrium solution concept is significant. The ultimate objective of several research, including \cite{C1,C20}, is to develop algorithms that converge to the Nash equilibrium of games. Technically speaking, confirming the Nash equilibrium's existence is the first stage in this series of studies. It is not difficult to guarantee that Nash equilibrium exists in the class of potential games. In fact, it is simple to demonstrate that there is a Nash equilibrium in this type of games with a compact action space and continuous payoff functions by looking at the potential function. 

Although there is already a strong mathematical foundation for studying exact one differential forms for potential games, describing potential games in general regardless of continuity of payoff functions still was open. As a result, characterization of aggregative potential games, has not received as much attention. Although the concept of Best Reply Potential Games in \cite{C22} is taken into consideration in the publication \cite{C5}, this definition may be debatable in and of itself. In fact, since the best response is a dynamic and dynamics often exist independently of games, the concept of inventing a new class of games based on best response correspondence may be misleading. Because such game may not be necessarily definable in term of payoff functions. Additionally, the existence of such a Best Reply Potential may be interesting for ensuring convergence of a specific dynamics, such as the best response or better response to a Nash equilibrium, but it may not be useful in some real-world situations where we need to know the form of this function to determine the game's Nash equilibrium or where we want to provide convergence analysis for another type of dynamics other than the better or best response.

We completely characterize the class of potential games in this study. To achieve this, we first examine the problem of describing potential games in the multidimensional case to comprehend how the potential function behaves in the multidimensional case where we have already discovered the potential function, keeping in mind that the the condition derived for potential games in terms of every pair of payoff functions given by \cite{monderer1996potential} only applies to the one dimensional case and one differential forms only applies to one dimensional variables. The next step is to gather some criteria for potential games in multidimensional action spaces. Then, we concentrate on the criterion and make it more simple for class of potential aggregative games.

In this paper, section \ref{Se_Nom} is dedicated to the notions and terminologies. We referred to the technical notations and some required definitions in this section. Section \ref{Prelim} is considered to mention some preliminary concepts regarding potential games in multi dimensional action spaces. This section follows by section \ref{Main-res} which is dedicated to the characterization of Potential Games and applications to aggregative games. Section \ref{sec:discussion} proposes a discussion on how the results can be applied to study the class of aggregative potential games. Finally, we conclude the paper in section \ref{sec:conclusion}.

\section{Notions and Terminology} \label{Se_Nom}
In this section we define some definitions and terminologies about the games which are going to be used in the rest of the manuscript. A game consists of $N$ agents represented by the set $\mathcal{N}:=\{1,\dots,N\}$. Each agent $i\in\mathcal{N}$ selects an action $x_i$ over a set of actions $K_i \in \mathbb{R}^n$ to minimize its payoff function $f_i:K_i \times K_{-i} \to \mathbb{R}$ where $K_{-i}:= \prod_{j \neq i} K_j$. A game $\Gamma$ is a tuple of agents $\mathcal{N}$, action space $K^N:=\prod_{i \in \mathcal{N}} K_i$, and payoff functions $f_i(\cdot)$ for $i \in \mathcal{N}$.
Class of potential games is defined in the following.
\begin{definition}[ Potential Games] \label{def_potential}
A game $\Gamma$ is an potential game, if there exists a function $f: K^N \rightarrow \mathbb{R}$ such that the following relation holds for all agents $i \in \mathcal{N}$,
\begin{equation}\label{eq_def_potential}
    f_i(x'_i,x_{-i})-f_i(x_i,x_{-i})=f(x'_i,x_{-i})-f(x_i,x_{-i})
\end{equation}
where $x_i' \in K_i$ and $x_i\in K_i$ and $x_{-i} \in K_{-i}$. The corresponding function $f: K^N \rightarrow \mathbb{R}$ is called the  potential function of the game $\Gamma$. 
\end{definition}

Finding the complete characterization of potential games  and the relationship between aggregative games with this class of potential games is the primary research questions of this work. In order to do this, we must also establish the idea of aggregative games. For this aim let $\bar{x}=\sum_{i=1}^N x_i$ denote the aggregate of all players decisions. We use $\bar{x}_{-i}$ to denote the aggregate of all players' decisions except for player $i$, i.e., $$\bar{x}_{-i}=\sum_{j=1,j\ne i}^N x_j.$$ 
Let us define the Minkowski sum of the sets $K_i$ with $\bar{K}$ as follows:
\begin{equation}\label{k_bar}
    \bar{K}\triangleq \sum_{i=1}^{N}K_i.
\end{equation}
and let $\bar{x}$ be the aggregate of players decisions $x_i$, i.e., 
\begin{equation}\label{x_bar}
    \bar{x}\triangleq \sum_{j=1}^N x_j=x_i+\bar{x}_{-i},\qquad \bar{x}\in \bar{K}.
\end{equation}
Having $\bar{x}_{-i}$, player $i$ is confronted with the following optimization problem:
\begin{eqnarray}\label{game}
    &&\min f_i(x_i,x_{-i})\triangleq \tilde{f}_i\big(x_i,g(\bar{x})\big), \nonumber\\ 
    && s.t.\ \ \  x_i \in K_i,
\end{eqnarray}
where $g:\bar{K}\rightarrow \mathbb{R}^m$ for some $m\in \mathbb{N}$. Therefore, following definition make sense for aggregative games. 
\begin{definition}[Aggregative Games] \label{aggregative_games}
Game $\Gamma=(\mathcal{N},\{\tilde{f}_i,K_i\}_{i\in\mathcal{N}})$ is aggregative.
\end{definition}
It should be noted that derivation of conditions obtained in \cite{monderer1996potential} for games to be potential, which is based on one differential forms \cite{Rudin-Analysis} are limited to single dimensional action spaces. We focus at the following section on the class of potential games in the multidimensional case. 

\section{Preliminaries}\label{Prelim}
This section look at the problem of conservative vector field through exact one-forms point of view. Let us write $\mathbf{a}=(a_1,a_2,\ldots,a_I)$.
 One-form $w_\mathbf{a}(v)$ in general at arbitrary point $\mathbf{a}$ in Euclidean space $\mathbb{R}^I$ are linear functional on the space of tangent vectors $v$ at the point $\mathbf{a}$. Considering tangent vector $d\mathbf{a}$ at point $\mathbf{a}$ in Euclidean space there is a unique function $F:\mathbb{R}^I\rightarrow\mathbb{R}^I$ such that
\begin{equation}\label{one-form}
    w_\mathbf{a}(d\mathbf{a})=F(\mathbf{a})\cdot d\mathbf{a}
\end{equation}

Following theorem give us precise conditions for a one-differential form to be exact over convex sets in $\mathbb{R}^I$. 
\begin{thm}\label{exact-form}
Let $u_i$ be some class $C^1$ functions on the convex set $E\subset \mathbb{R}^I $ for some $I\in \mathbb{N}$, then 1-form $\omega=\sum_{i=1}^{I} u_i(\mathbf{a})da_i$ for $i \in \{1,2,\ldots,I\}$ for $\mathbf{a}\in E$ is exact if and only if we have 
\begin{equation}\label{nec-suf-exact}
    D_j u_i(\mathbf{a})=D_i u_j(\mathbf{a}),
\end{equation}
for all $i,j \in \{1,2,\ldots,I\}$, where $D_i$ is the partial derivative with respect to $a_i$.
\end{thm}
\begin{proof}
The proof of this theorem immediately follows from Remark 10.35 a) and Theorem 10.39 of Rudin's book \cite{Rudin-Analysis}.
\end{proof}
Considering this fact that for vector field $F:\mathbb{R}^I\rightarrow\mathbb{R}^I$ we can uniquely define one-form as it is introduced in \eqref{one-form}, one can introduce a one-form for each arbitrary game by viewing the concatenation of derivatives of cost function of each agent with respect to its own decision variable as a vector field. For example, for a game $\Gamma=(\mathcal{N},\{K_i,f_i\}_{i\in\mathcal{N}})$ with one dimensional action space one can consider vector field $$G(x)=(\frac{\partial f_1}{\partial x_1},\frac{\partial f_2}{\partial x_2},\ldots,\frac{\partial f_N}{\partial x_N}),$$
and assign one-form $G(x)\cdot dx$ to this game $\Gamma$. Based on the definition of the exact one-form if the one-form corresponding to this game is exact, the game is potential because one can write $G(x)\cdot dx= d\phi$ for some scalar function $\phi:\mathbb{R}^I\rightarrow\mathbb{R}$.
From now on, whenever we use the term curve, we denote oriented curve with parameterization $t$ unless otherwise specified. Considering $K_i\subset \mathbb{R}^n$ we define 1-differential forms in the extended action space in the sequel. Let us consider curve 
\begin{align*}
    \gamma_i(t)= &\big(\gamma_{11}(t_0),\ldots,\gamma_{1n}(t_0),\gamma_{i1}(t),\ldots,\nonumber \\
    &\gamma_{in}(t),\gamma_{N1}(t_0),\ldots,\gamma_{Nn}(t_0)\big)\in \mathbb{R}^{Nn},
\end{align*}
for all $i\in \mathcal{N}$, $t\in \mathbb{R}$, and some $t_0\in \mathbb{R}$. Clearly, for all $i,j\in \mathcal{N}$ we can write $\gamma_i(t_0)=\gamma_j(t_0)$. We call this point as $\gamma(t_0)$.

In the following we describe Potential Games in parametric sense. For arbitrary curves $\gamma_i(t)$ for all $i\in \mathcal{N}$ and $t\in \mathbb{R}$, game $\Gamma$ is potential if there is a scalar function $\phi$ such that 
\begin{equation}\label{eq-MDAS-pot}
    \frac{df_i\big(\gamma_i(t)\big)}{dt}=\frac{d\phi\big(\gamma_i(t)\big)}{dt}
\end{equation}
for all $i\in \mathcal{N}$ and $t\in \mathbb{R}$. We call $\phi$ as potential function. Latter term describes a parametric version of Potential Games. In the sequel of the manuscript we may refer to these view points alternatively. 
\begin{lemma}
A game is potential game if and only if for all $i\in \mathcal{N}$ and $m \in \{1,2,\ldots,n\}$ we have $\frac{\partial f_i}{\partial a_{im}}=\frac{\partial \phi}{\partial a_{im}}$ for some scalar function $\phi$.
\end{lemma}
\begin{proof}
If a game is potential game, defining $\gamma_{il}(t)=\gamma_{il}(t_0)$ for all $l\in \{1,2,\ldots,n\}/\{m\}$, according to the \eqref{eq-MDAS-pot} we have $\frac{\partial f_i}{\partial a_{im}}d\gamma_{im}=\frac{\partial \phi}{\partial a_{im}}d\gamma_{im}$. Therefore, because $\gamma_{im}$ is selected arbitrarily, we can conclude that $\frac{\partial f_i}{\partial a_{im}}=\frac{\partial \phi}{\partial a_{im}}$. For the other way around we know that since for all $i\in \mathcal{N}$ and $m \in \{1,2,\ldots,n\}$ we have $\frac{\partial f_i}{\partial a_{im}}=\frac{\partial \phi}{\partial a_{im}}$, then we considering the partial derivatives we can simply write $$\frac{df_i\big(\gamma_i(t)\big)}{dt}=\sum_{m=1}^n\frac{\partial f_i}{\partial a_{im}}\cdot\frac{d\gamma_{im}}{dt}=\sum_{m=1}^n\frac{\partial \phi}{\partial a_{im}}\cdot\frac{d\gamma_{im}}{dt}=\frac{d\phi\big(\gamma_i(t)\big)}{dt}.$$ This completes the proof.
\end{proof}

\begin{thm}\label{Th-MDASPG}
Let us consider $K_i\subset \mathbb{R}^n$ for all $i\in \mathcal{N}$. Therefore, the decision variable is of the following form $$ \mathbf{a}=(a_{11},\ldots,a_{1n},\ldots,a_{i1},a_{i2},\ldots,a_{in},\ldots,a_{N1},\ldots,a_{Nn}),$$
for all $i\in \mathcal{N}$. Assuming $K^N$ is convex, then
\begin{align*}
    &\big(\frac{\partial f_1}{\partial a_{11}},\ldots,\frac{\partial f_1}{\partial a_{1n}},\ldots,\frac{\partial f_N}{\partial a_{N1}},\ldots,\frac{\partial f_N}{\partial a_{Nn}}\big)\nonumber \\
    &=\big(\frac{\partial \phi}{\partial a_{11}},\ldots,\frac{\partial \phi}{\partial a_{Nn}}\big),
\end{align*}
and the game is potential game if and only if
\begin{equation}\label{nec-suf-cond}
    \frac{\partial^2 f_i}{\partial a_{ip}\partial a_{jq}}=\frac{\partial^2 f_j}{\partial a_{jq}\partial a_{ip}},
\end{equation}
for all $i,j \in \mathcal{N}$ and $p,q \in \{1,2,\ldots,n\}$, for some scalar function $\phi$.
\end{thm}

\begin{proof}
The order of the decision variables is important in calculating any $f_i,\ i\in \mathcal{N}$. Thus, considering $M=Nn$, let us consider this decision vector in the form of
\begin{align*}
    \mathbf{a}&=(a_{11},\ldots,a_{1n},\ldots,a_{i1},a_{i2},\ldots,a_{in},\ldots,a_{N1},\ldots,a_{Nn})\nonumber \\
    &=(y_1,y_2,\ldots,y_{Nn})=\mathbf{y},
\end{align*}

Differential form 
$$\gamma=\sum_{s=1}^{Nn} \frac{\partial u_s(\mathbf{y})}{\partial y_{s}}dy_{s},$$ where $u_{(l-1)n+r}=f_l$ for all $l\in \mathcal{N}$ and $r\in \{1,2,\ldots,n\}$. 
According to (\ref{nec-suf-exact}) for the convex action space $K^N$ the differential form $\gamma$ is exact if and only if $D_l u_d(\mathbf{y})=D_d u_l(\mathbf{y})$ for all $l,d \in \{1,2,\ldots,Nn\}$. This will lead to this fact that game $\Gamma$ is potential if and only if $$ \frac{\partial^2 f_i}{\partial a_{ip}\partial a_{jq}}=\frac{\partial^2 f_j}{\partial a_{jq}\partial a_{ip}},$$
for all $i,j \in \mathcal{N}$ and $p,q \in \{1,2,\ldots,n\}$. This completes the proof.
\end{proof}

\section{Characterization of Potential Games}\label{Main-res}
As a tool to understand other classes, let us take a look at the Potential Functions from another prospective. According to \eqref{eq-MDAS-pot} for Potential Games, for every $i\in\mathcal{N}$ by integrating both hand side of the equation and using the stokes theorem we can write
\begin{equation}\label{char1-pot}
   \phi\big(\gamma_i(t_0+\epsilon)\big)-\phi\big(\gamma_i(t_0)\big)=f_i\big(\gamma_i(t_0+\epsilon)\big)-f_i\big(\gamma_i(t_0)\big).
\end{equation}
In the non parametric form we can alternatively state for every $z\in K^N$ and $y_i\in K$ that 
\begin{equation}\label{char1-pot-non-para}
   \phi(z_i+y_i,z_{-i})-\phi(z_i,z_{-i})=f_i(z_i+y_i,z_{-i})-f_i(z_i,z_{-i}).
\end{equation}\label{Th-char2-pot-non-para}
This latter definition of potential function and potential games indeed does not need any notions of differentiability and continuity. Following theorem provides some necessary condition for the format of Potential Function of the Potential Games. 

\begin{thm}\label{theorem-3}
If game $\Gamma$ is a Potential Game, then the Potential Function $\phi$ satisfies 
\begin{align}\label{char2-pot-non-para}
    \phi(z+y)-\phi(z)&=\sum_{i=1}^N\big(f_i(z_1+y_1,\ldots,z_i+y_i,z_{i+1},\ldots,z_n)\nonumber \\
    &-f_i(z_1+y_1,\ldots,z_{i-1}+y_{i-1},z_{i},\ldots,z_n)\big).
\end{align}
\end{thm}
\begin{proof}
 Considering a path $P:(z)\rightarrow(z_1+y_1,z_{-1})\rightarrow(z_1+y_1,z_2+y_2,z_{-\{1,2\}})\rightarrow \ldots\rightarrow(z+y)$ and writing \eqref{char1-pot-non-para} for every two sequential components of this path and adding up them all together we have \eqref{char2-pot-non-para}. This completes the proof.
\end{proof}
\begin{rem}
The right hand side of \eqref{char2-pot-non-para} always exist. Hence, being able to write \eqref{char2-pot-non-para}, does not mean that game $\Gamma$ is Potential. A sufficient condition for being so, however, is that there exists such non constant function $\phi$ satisfying \eqref{char2-pot-non-para}.
\end{rem}
We choose this particular path $P$ because it might be useful in the study of aggregative potential games. Prior to move to the further analysis of potential games, we explore some characteristics of the right hand side of \eqref{char2-pot-non-para}. Let us consider the right hand side of \eqref{char2-pot-non-para} as $h_P(y,z)$ where $h_P:\mathbb{R}^{2Nn}\rightarrow\mathbb{R}$ and $P$ stands for the particular path described above. 
\begin{definition}[Abnormal Game]\label{ab-game} A game $A=(\mathcal{N},\{f_i,K_i\}_{i\in\mathcal{N}})$ is abnormal game if there is an $i\in\mathcal{N}$ such that for every $x_{-i}\in K_{-i}$ and for every $x_{i}\in K_{i}$ we have $f_i(x_i,x_{-i})=C_i(x_{-i})$ for some real function $C_i:K_{-i}\rightarrow \mathbb{R}$.
\end{definition}
This definition states that in abnormal games there is a person whose action is not affecting her payoff function but it may affect other's cost functions. In this case for this person there is no incentive to make a decision with respect to other players decision. In potential games which are abnormal, the potential function is not sensitive to the decision variable of some agents. Aggregative game can't be an abnormal game due its definition.
Let us consider $y,z\in K^N$
\begin{itemize}
    \item Property 1. If $h_p(y,z)=h_p(y+z,0)-h_p(z,0)$ then, we can write the following expression $h_p\big((y_{-N},0),(z_{-N},0)\big)=h_p\big((y_{-N},0)+(z_{-N},0),0\big)-h_p\big((z_{-N},0),0\big)$. 
    \item Property 2. For aggregative games, the function $h_P(z,0)$ can't be a zero function.
\end{itemize}

\begin{rem}
Since $h_P(0,0)=0$ this function $h_P(z,0)$ can't be constant function apart from $0$.
\end{rem} 

Following theorem illustrates why the Property 2, holds true for aggregative games.
\begin{thm}\label{prop-2}
In aggregative game $\Gamma=(\mathcal{N},\{\tilde{f}_i,K_i\}_{i\in\mathcal{N}})$ the function $h_P(z,0)$ is a non zero function. 
\end{thm}
\begin{proof}
We prove this statement via contradiction. Let us consider $z=(0,\ldots,0,u,v)$ and assume $h_P(x,0)$ is zero function for all $x\in K^N$. We have
\begin{align}\label{proof-prop-2}
    h_P(z,0)&=\tilde{f}_{N-1}\big(u,g(u)\big)-\tilde{f}_{N-1}\big(0,g(0)\big)\nonumber \\
    &+\tilde{f}_{N}\big(v,g(u+v)\big)-\tilde{f}_{N}\big(0,g(u)\big)=0.
\end{align}
Considering $v=0$ in \eqref{proof-prop-2}, we find that 
\begin{equation}\label{proof-prop-2-1}
    \tilde{f}_{N-1}\big(u,g(u)\big)-\tilde{f}_{N-1}\big(0,g(0)\big)=0
\end{equation}
Thus we have 
\begin{equation}\label{proof-prop-2-2}
    \tilde{f}_{N}\big(v,g(u+v)\big)=\tilde{f}_{N}\big(0,g(u)\big).
\end{equation}
This implies that for every $x_{-N}$ in this aggregative game, $f_N(x_N,x_{-N})=f_N(0,x_{-N})$ which is aligned with the definition of abnormal games for player $N$. This yields a contradiction and complete the proof.
\end{proof}

\begin{rem}
Excluding aggregative games, in the general case, there might be non abnormal games which are satisfying \eqref{proof-prop-2}, \eqref{proof-prop-2-1} and \eqref{proof-prop-2-2}. For instance, consider for all $i\in \mathcal{N}$ that $f_i(x)=\prod_{j=1}^N x_j$.
\end{rem}

In the following theorem we provide a necessary and sufficient condition for a game to be potential.
\begin{thm}\label{nec-suff-cond-agg-pot}
 The game $G=(\mathcal{N},\{f_i,K_i\}_{i\in\mathcal{N}})$ is Potential Game if and only if
 \begin{align}\label{char2-agg-pot-non-para}
    &W(z+y)-W(z)\nonumber \\
    &=\sum_{i=1}^N\big(f_i(z_1+y_1,\ldots,z_i+y_i,z_{i+1},\ldots,z_n)\nonumber \\
    &-f_i(z_1+y_1,\ldots,z_{i-1}+y_{i-1},z_{i},\ldots,z_n)\big).
\end{align}
for some scalar function $W$. Moreover, $W$ is the Potential Function.
\end{thm}
\begin{proof}
If the game $G$ is Potential Game then, considering the Potential Function $\phi$, using \eqref{char2-pot-non-para} we have \eqref{char2-agg-pot-non-para} with $W=\phi$. For the other way around, existence of some non constant scalar function $W$ such that \eqref{char2-agg-pot-non-para} holds, we can check for $y=(0,\ldots,y_i,\ldots,0)$ and arbitrary $z$ and show for every $j\in \mathcal{N}$ that $W(z_i+y_i,z_{-i})-W(z_i,z_{-i})=f_i\big(z_i+y_i,z_{-i}\big)-f_i\big(z_i,z_{-i}\big)$. Hence, The game is Potential Game with Potential Function $W$. This completes the proof.
\end{proof}

Following theorem states an alternative necessary and sufficient condition for a game to be Potential Game.

\begin{thm}\label{nec-suff-cond-2-agg-pot}
Suppose that $K_i$ is symmetric (i.e. if $x_i\in K_i$ then $-x_i\in K_i$ as well) and $0\in K_i$ for all $i \in \mathcal{N}$. The game $G=(\mathcal{N},\{f_i,K_i\}_{i\in\mathcal{N}})$ is Potential Game if and only if $h_P(y,z)=h_P(y+z,0)-h_P(z,0)$ and the potential function is $C-h_P(-z,z)$, where $C$ is some constant.
\end{thm}
\begin{proof}
We start with $h_P(y,z)=\phi(z+y)-\phi(z)$. Then we have $h_P(-z,z)=\phi(0)-\phi(z)$. As a result, we have $h_P(y,z)=h_P(-z,z)-h_P\big(-(z+y),(z+y)\big)$ Additionally, by considering $z=0$ in the latter term, we can write $h_P(y,0)=h_P(0,0)-h_P(-y,y)$ and substituting $y=0$ in this latter equation we have $h_P(0,0)=0$. We can also write $h_P(0,z)=0$ From this expressions we can also write 
\begin{equation}\label{final-eq}
    h_P(y,z)=h_P(y+z,0)-h_P(z,0).
\end{equation}
\end{proof}

This question may raise naturally that what relation between potential games and duopoly potential games holds. In the sequel we investigate this relation and prove that the necessary and sufficient condition for $N$-player game to be potential is that every $2$-player of this game with some arbitrary vector of other decision variables which are not involved in that sub game, forms a potential game.

For every arbitrary $i,j\in \mathcal{N}$, $y_j,z_j\in K_j$, $y_i,z_i\in K_i$, and $z_{-\{i,j\}}\in K_{-\{i,j\}}$  let us define $h_{ij}(y_j,y_i,z_j,z_i,;z_{-\{i,j\}})$ as follows

\begin{align}\label{char-agg-pot-def-1}
    h_{ij}(y_j,y_i,z_j,z_i,;z_{-\{i,j\}})&=f_{i}\big(z_i+y_i,z_j;z_{-\{i,j\}}\big)-f_{i}\big(z_i,z_j;z_{-\{i,j\}}\big)\nonumber \\
    &+f_{j}\big(z_j+y_j,z_i+y_i,z_{-\{i,j\}}\big)-f_{j}\big(z_j,z_i+y_i,z_{-\{i,j\}}\big).
\end{align}

Consider a path $P^*$ such that last two changes take place in $i^*$ and $j^*$, respectively. Let us assume $y=(0,\ldots,y_{i^*},0,\ldots,y_{j^*},0,\ldots)$ and $z=(z_1,\dots,z_N)$ such that $z_{-\{i^*,j^*\}}=z^*_{-\{i^*,j^*\}}$. We have in potential games that
\begin{align}\label{theorem-item-3-1}
    &h_{i^*j^*}\big(y_{j^*},y_{i^*},z_{j^*},z_{i^*};z^*_{-\{i^*,j^*\}}\big)=h_{P^*}(y,z)=h_{P^*}(y+z,0)-h_{P^*}(z,0),
\end{align}
for the right hand side of \eqref{theorem-item-3-1} we can write
\begin{align}\label{theorem-item-3-1-2}
    &h_{P^*}(y+z,0)-h_{P^*}(z,0)=h_{i^*j^*}\big(y_{j^*}+z_{j^*},y_{i^*}+z_{i^*},0,0;z^*_{-\{i^*,j^*\}}\big)-h_{i^*j^*}\big(z_{j^*},z_{i^*},0,0;z^*_{-\{i^*,j^*\}}\big).
\end{align}

For the finite path $\mathcal{Q}=(q^0,\ldots,q^L)$, where $q^l \in K^N$ for $l\in\{0,1,\ldots,L\}$, and for a vector $f=(f_1,\ldots,f_N)$ of functions $f_i:K^N\rightarrow \mathbb{R}$, one can define

\begin{equation}
    I(\mathcal{Q},f)=\sum_{e=1}^L\big(f_{i_e}(q_e)-f_{i_e}(q_{e-1})\big),
\end{equation}
where, $i_e$ is the unique deviator at step $e$ (i.e., $q_{e}^{i_e}\neq q_{e}^{i_{e-1}}$). The path $\mathcal{Q}$ is called close if $q_0=q_N$. Moreover, it is simple closed path if it does not cross itself at any intermediate point $l\in\{0,1,\ldots,L-1\}$. The length of a simple closed path is the number of different vertices of this path.
It is proven that game $\mathcal{F}=(\mathcal{N},\{f_i,K_i\}_{i\in\mathcal{N}})$ is potential if and only if for every finite simple closed path $\mathcal{Q}$ of length $4$, $I(\mathcal{Q},f)=0$ \cite{monderer1996potential}.

\begin{thm}\label{theorem 7}
The $N$-player game $\Gamma=(\mathcal{N},\{f_i,K_i\}_{i\in\mathcal{N}})$ is potential if and only if for all $i^*,j^*\in \mathcal{N}$ and $z^*_{-\{i^*,j^*\}}\in K_{-\{i^*,j^*\}}$ we have 
\begin{align}\label{char-agg-pot-def-0}
    &h_{i^*j^*}(y_{j^*},y_{i^*},z_{j^*},z_{i^*};z^*_{-\{i^*,j^*\}})=h_{i^*j^*}(y_{j^*}+z_{j^*},y_{i^*}+z_{i^*},0,0;z^*_{-\{i^*,j^*\}})-h_{i^*j^*}(z_{j^*},z_{i^*},0,0;z^*_{-\{i^*,j^*\}}).
\end{align}
\end{thm}

\begin{proof}
Let $N$-player game $\Gamma=(\mathcal{N},\{f_i,K_i\}_{i\in\mathcal{N}})$ is potential. Hence, for every path $\mathcal{Q}$ of length $4$ we have $I(\mathcal{Q},f)=0$. Let 
\begin{align*}
    \mathcal{Q}:(z_1,z_2,\ldots,z_N)&\rightarrow(z_1,\ldots,z_{i^*}+y_{i^*},\ldots,z_N) \nonumber \\
    &\rightarrow (z_1,\ldots,z_{i^*}+y_{i^*},\ldots,z_{j^*}+y_{j^*},\ldots,z_N)\nonumber \\
    &\rightarrow(z_1,\ldots,z_{i^*},\ldots,z_{j^*}+y_{j^*},\ldots,z_N)\nonumber \\
    &\rightarrow(z_1,\ldots,z_{i^*},\ldots,z_{j^*},\ldots,z_N).
\end{align*} 
Moreover, let us consider $Z=(z_1,z_2,\ldots,z_N)$ and $Y=(0,\ldots,y_{i^*},\ldots,y_{j^*},\ldots,0)$. Because of the fact that the game $\Gamma$ is potential then \eqref{theorem-item-3-1} holds and since \eqref{theorem-item-3-1-2} holds in general, we can conclude that \eqref{char-agg-pot-def-0} holds true. For the other way around, let assume \eqref{char-agg-pot-def-0} holds for every $i^*,j^*\in \mathcal{N}$ and $z^*_{-\{i^*,j^*\}}\in K_{-\{i^*,j^*\}}$, then for arbitrary simple closed path of length $4$ of $\mathcal{Q}$ we have 
\begin{align}\label{char-agg-pot-theorem-2}
    I(\mathcal{Q},f)&=h_{i^*j^*}(0,y_{i^*},z_{j^*},z_{i^*};z_{-\{i^*,j^*\}})\nonumber\\
    &+h_{i^*j^*}(y_{j^*},0,z_{j^*},z_{i^*}+y_{i^*};z_{-\{i^*,j^*\}})\nonumber\\
    &+h_{i^*j^*}(0,-y_{i^*},z_{j^*}+y_{j^*},z_{i^*}+y_{i^*};z_{-\{i^*,j^*\}})\nonumber\\
    &+h_{i^*j^*}(-y_{j^*},0,z_{j^*}+y_{j^*},z_{i^*};z_{-\{i^*,j^*\}})\nonumber\\
    &=0,
\end{align}
according to multiple time applying of \eqref{char-agg-pot-def-0}. Since $\mathcal{Q}$ is arbitrary path of length $4$, in the action space for the $N$-player game, this game is potential. This completes the proof.
\end{proof}

Following example is designed to verify the above theorem through the 3-player Cournot game.

\textbf{Example 1.} Suppose that we have following utility functions:
\begin{align}
    &f_1(x_{1},x_{2},x_{3})= (A-B\bar{x})x_1-Cx_1, \\
    &f_2(x_{1},x_{2},x_{3})= (A-B\bar{x})x_2-Cx_2, \\
    &f_3(x_{1},x_{2},x_{3})= (A-B\bar{x})x_3-Cx_3.
\end{align}
We choose players $1,2$ and rewrite the \eqref{char-agg-pot-def-0} to check if it holds. We have
\begin{align}\label{EXM-2-1}
    h_{12}(y_{2},y_{1},z_{2},z_{1};z_3)&=(A-B(\bar{z}+y_1))(z_1+y_1)-C(z_1+y_1)\nonumber \\
    &-\big((A-B\bar{z})z_1-Cz_1\big)\nonumber \\
    &+(A-B(\bar{z}+y_1+y_2))(z_2+y_2)-C(z_2+y_2)\nonumber \\
    &-\big((A-B(\bar{z}+y_1))z_2-Cz_2\big).
\end{align}

Moreover, we have
\begin{align}\label{EXM-2-2}
   &h_{i^*j^*}(y_{j^*}+z_{j^*},y_{i^*}+z_{i^*},0,0;k^*_{-\{i^*,j^*\}})-h_{i^*j^*}(z_{j^*},z_{i^*},0,0;k^*_{-\{i^*,j^*\}})\nonumber \\
   &=(A-B(z_1+z_3+y_1))(z_1+y_1)-C(z_1+y_1)\nonumber \\
   &+(A-B(\bar{z}+y_1+y_2))(z_2+y_2)-C(z_2+y_2)\nonumber \\
   &-\big((A-B(z_1+z_3))(z_1)-C(z_1)\big)\nonumber \\
   &-\big((A-B(\bar{z}))(z_2)-C(z_2)\big).
\end{align}
We can verify that \eqref{EXM-2-1} and \eqref{EXM-2-2} are equivalent. This means the game is potential.

Since the condition of being potential can be boiled down to check some equality condition in terms of $h_{ij}$ for every $i,j \in \mathcal{N}$, we might be able to find the potential function of a $N$-player potential games in terms of these functions. Following theorem provide a characterization of potential function of $N$-player potential games. Prior to proceeding with the theorem we define $\hat{v}_i$  by $(v_1,\ldots,v_i,0,\ldots,0)$. Let consider path $P$ in theorem \ref{theorem-3}, and $P'$, $P''$ the first 3 and the first 2 steps of path $P$ respectively, and $P^*$ as defined previously, then

\begin{thm}\label{theorem 8}
The $N$-player game $\Gamma=(\mathcal{N},\{f_i,K_i\}_{i\in\mathcal{N}})$ is potential with potential function $\phi$ then if $N=2K+1$ for some $K\in \mathbb{N}$ we have
\begin{align}\label{theorem 8-1}
    \phi(0)&+h_{P'}((z_1,z_2,z_3),0)+\sum_{i=2}^K h_{2i,2i+1}(z_{2i+1},z_{2i},0,0;\hat{z}_{2i-1})=\phi(z),
\end{align}
else if $N=2K$ we have
\begin{align}\label{theorem 8-2}
    \phi(0)&+h_{P''}((z_1,z_2),0)+\sum_{i=2}^K h_{2i-1,2i}(z_{2i},z_{2i-1},0,0;\hat{z}_{2i-2})=\phi(z).
\end{align}
\end{thm}
\begin{proof}
We prove the statement for $N=2K+1$ for some $K\in \mathbb{N}$. The proof for the other case is identical to latter case.
Let assume in path $P^*$, $i^*=2K$ and $j^*=2K+1$. Utilizing \eqref{theorem-item-3-1} with $y=(0,\ldots,z_{N-1},z_{N})$ and $z=(z_1,\ldots,z_{N-2},0,0)$, we have
\begin{align}\label{theorem 8-3}
    h_{2K,2K+1}\big(z_{2K+1},z_{2K},0,0;\hat{z}_{\{2K-1\}}\big)&=h_{P^*}(y,z)\nonumber \\
    &=h_{P^*}((z_1,\ldots,z_{N-2},z_{N-1},z_{N}),0)\nonumber \\
    &-h_{P^*}((z_1,\ldots,z_{N-2},0,0),0).
\end{align}
Additionally, we know that by theorem \ref{nec-suff-cond-2-agg-pot} we have
\begin{align}\label{theorem 8-4}
    &\phi(z_1,\ldots,z_{N-2},z_{N-1},z_{N})=\phi(0)+h_{P^*}((z_1,\ldots,z_{N-2},z_{N-1},z_{N}),0).
\end{align}
Rearranging \eqref{theorem 8-3} and substituting in \eqref{theorem 8-4} we have
\begin{align}\label{theorem 8-5}
    \phi(z_1,\ldots,z_{N-2},z_{N-1},z_{N})&=\phi(0)+h_{P^*}((z_1,\ldots,z_{N-2},0,0),0)\nonumber \\
    &+h_{N-1,N}\big(z_{N},z_{N-1},0,0;\hat{z}_{\{N-2\}}\big).
\end{align}
In \eqref{theorem 8-5} since $N-2$ is again an odd number and the $N-2$-player game which is obtained by omitting players $N$ and $N-1$ from game $\Gamma$, and since conditions of theorem \ref{theorem 7} is held for this $N-2$-player game, it is potential and we can repeat the argument. Continuing this process we eventually reach \eqref{theorem 8-1}. This completes the proof.
\end{proof}

\textbf{Example 2.} Suppose that we have following utility functions:
\begin{align}
    &f_1(x_{1},x_{2},x_{3},x_{4})= (A-B\bar{x})x_1-Cx_1, \\
    &f_2(x_{1},x_{2},x_{3},x_{4})= (A-B\bar{x})x_2-Cx_2, \\
    &f_3(x_{1},x_{2},x_{3},x_{4})= (A-B\bar{x})x_3-Cx_3,\\
    &f_4(x_{1},x_{2},x_{3},x_{4})= (A-B\bar{x})x_4-Cx_4.
\end{align}
Using theorem \ref{theorem 8}, we have
\begin{align}\label{example-2}
\phi(x)=\phi(0)&+(A-Bx_1)x_1-Cx_1\nonumber \\
&+(A-B(x_1+x_2))x_2-Cx_2\nonumber \\
&+(A-B(x_1+x_2+x_3))x_3-Cx_3\nonumber \\
&+(A-B(x_1+x_2+x_3+x_4))x_4-Cx_4.
\end{align}
We can see that $\phi(x_1+y_1,x_2,x_3,x_4)-\phi(x_1,x_2,x_3,x_4)=(A-B(\bar{x}))y_1-Bx_1y_1-By_1^2-Cy_1$. It is not difficult to show that  $f_1(x_1+y_1,x_2,x_3,x_4)-f_1(x_1,x_2,x_3,x_4)$ is equal to the same value.
\section{Discussion}\label{sec:discussion}
Theorem \ref{theorem 7} states that for a game to be potential, for every point $z^*_{-\{i,j\}}$ in the action space $K_{-\{i,j\}}$ there is a function equation must be hold, and vice versa, if the function equation holds in every point $z^*_{-\{i,j\}}\in K_{-\{i,j\}}$, the game must be potential. This function equation obtained in theorem \ref{theorem 7} is comparable with \eqref{nec-suf-cond} for games with smooth payoff functions. In other words, this latter condition for games with non continuous payoff functions can be extended to the theorem \ref{theorem 7} of the current study. It should be also noted that by looking at conditions stated in theorem \ref{nec-suff-cond-2-agg-pot}, at the first look it seems the action space needs to be a box which is symmetric with respect to its center, to be able to proceed with the rest of theorems built upon theorem \ref{nec-suff-cond-2-agg-pot}. However, in case of constrained action space as long as payoff functions are non infinity along path $\mathcal{Q}$ introduced in \ref{theorem 7}, we can repeat all the results for this case as long as $x$ and $x+y$ are in the constrained action space by focusing on a hypercube containing the entire constrained action space. Therefore, for the problem with constrained action space theorem \ref{theorem 7} provides a necessary condition. Additionally, due to its nature, the condition in this theory is sufficient because then we can show $I(\mathcal{Q},f)=0$ in the hypercube containing the constrained action space. 

In aggregative games instead of $z^*_{-\{i,j\}}$, the term $\bar{z}^*_{-\{i,j\}}$ appears in the derivations. Therefore, the conditions in theorem \ref{theorem 7}, in particular \eqref{char-agg-pot-def-0} needs only to be satisfied for every constant $\bar{k}^*_{-\{i,j\}} \in \bar{K}_{-\{i,j\}}$. Due to the fact that in aggregative games, from point of view of agent $i,j$ it does not make any difference, in unconstrained problem (with unbounded action space) one can assign $\bar{k}^*_{-\{i,j\}}$ to action of one player apart from $i,j$, let say $l\in \mathcal{N}$. As a result, one can geometrically interpret the condition of theorem \ref{theorem 7} to this that the condition \eqref{char-agg-pot-def-0} only needs to be held for all $i,j \in \mathcal{N}$ along 
\begin{enumerate}
    \item direction defined by decision variable of player $l^{(1)}$ in the action space where non of $i,j$ are $l^{(1)}$,
    \item direction defined by decision variable of player $l^{(2)}$ in the action space where one of $i,j$ is $l^{(1)}$,
    \item direction defined by decision variable of player $l^{(3)}$ in the action space where one of $\{i,j\}$ are $\{l^{(1)},l^{(2)}\}$,
\end{enumerate}
where, $l^{(1)},l^{(2)},l^{(3)}$ are orthogonal axes.
\section{Conclusions}\label{sec:conclusion}
In this paper we focused on class of potential games and derived a necessary and sufficient condition for games to fall under this class of games. We stepped further and simplified the general criteria we obtained for potential games for class of aggregative games. This relation completely describes aggregative potential games in sense of every two player's payoff functions coupling behavior. We checked the condition through a $3$-player Cournot game. We also examined the form of potential function for potential games through an example of $4$-player aggregative games.

\end{document}